\documentclass{amsart}
\usepackage{a4wide}
\usepackage{amssymb}
\usepackage{amsmath}
\usepackage{amsfonts}
\usepackage{amsthm}
\usepackage{graphicx}
\usepackage{epsfig}
\usepackage{longtable}
\usepackage[usenames,dvipsnames]{color}
\usepackage{hyperref}
\newtheorem{thm}{Theorem}
\newtheorem{lem}[thm]{Lemma}
\newtheorem{cor}[thm]{Corollary}

\theoremstyle{definition}

\theoremstyle{remark}

\newcommand{\C}{\mathbb{C}}

\newcommand{\DEF}{{:=}}

\newcommand{\Cset}{\mathbb{C}}

\newcommand{\PT}[1]{\mathbf{#1}}

\DeclareMathOperator{\dd}{\mathrm{d}}

\DeclareMathOperator{\BellB}{B}

\DeclareMathOperator{\gammafcn}{\Gamma}

\DeclareMathOperator{\PolyLog}{L}

\DeclareMathOperator{\Stirlings}{s}

\DeclareMathOperator{\zetafcn}{\zeta}

\DeclareMathOperator{\HyperF}{F}

\newcommand{\Hypergeom}[5]{{\sideset{_#1}{_#2}\HyperF\!\left(\substack{\displaystyle#3\\\displaystyle#4};#5\right)}}

\newcommand{\fallingFactorial}[2]{{\left[#1\right]_{#2}}}
\newcommand{\Pochhsymb}[2]{{\left(#1\right)_{#2}}}
\newcommand{\Eulerian}[2]{{\left\langle #1 \atop #2 \right\rangle}}
\newcommand{\usgnStirlingS}[2]{{\left[ #1 \atop #2 \right]}}

\allowdisplaybreaks[1]

\title[Explicit formulas for the Riesz energy of the $N$th roots of unity]{Explicit formulas for the Riesz energy of \\ the $N$th roots of unity}
\author{J. S. Brauchart} 
\thanks{\noindent The research of this author was supported by the Austrian Science Fund FWF projects F5510 (part of the Special Research Program (SFB) ``Quasi-Monte Carlo Methods: Theory and Applications''). The author also acknowledges the support of the Erwin Schr{\"o}dinger Institute in Vienna, where part of the work was carried out.}

\date{\today}

\begin{document}

\address{J.~S.~Brauchart:
Institut f\"ur Analysis und Computational Number Theory,
Technische Universit\"at Graz,
Steyrergasse 30,
8010 Graz,
Austria
}
\email{j.brauchart@tugraz.at}

\begin{abstract}
The paper Brauchart, Hardin and Saff [Bull. Lond. Math. Soc. 41(4) (2009)] gives the complete  asymptotic expansions of the Riesz $s$-energy of the $N$th roots of unity which form a universally optimal distribution of points on the unit circle in the sense of Cohn and Kumar [J. Amer. Math. Soc. 20 (2007)]. Here, exact formulas (valid for all $N \geq 2$) are obtained for the case when $s$ is an even integer. 
In the case of the singular Riesz $s$-potential $1/r^s$, $r$ the Euclidean distance between two points, a continuous modified energy approximation of the Riesz energy is used. Stirling numbers of the first kind, Eulerian numbers and special values of partial Bell polynomials play a central role. Several identities between these quantities are shown.
\end{abstract}


\maketitle

\centerline{Dedicated to Edward B. Saff on the occasion of his 70\textsuperscript{th} birthday.}

\section{Introduction}

For real $s \neq 0$, the Riesz $s$-energy of an $N$-point configuration $X_N = \{ \PT{x}_1, \dots, \PT{x}_N \}$ in the Euclidean space $\mathbb{R}^p$, $p \geq 1$, is given by 
\begin{equation} \label{eq:energy.functional}
\mathrm{E}_s( X_N ) = \mathop{\sum_{j=1}^N \sum_{k=1}^N}_{j \neq k} \frac{1}{\left| \PT{x}_j - \PT{x}_k \right|^s}.
\end{equation}
A fundamental question concerns the asymptotic expansion of the minimal (if $s > 0$) or maximal (if $s < 0$) Riesz $s$-energy of $N$ points restricted to an infinite compact subset $A \subset \mathbb{R}^p$ as $N \to \infty$; see \cite{BoHaSaBook, BrGr2014arXiv, Gr2014, HaSa2004, SaKu1997, SaTo1997}. 
%
%
In the \emph{potential-theoretic case}, when the energy integral
\begin{equation} \label{eq:energy.integral}
\int \int \frac{1}{| \PT{x} - \PT{y} |^s} \, \dd \mu( \PT{x} ) \dd \mu( \PT{y} )
\end{equation}
is finite for at least one positive measure supported on $A$, the leading-term behaviour follows from arguments of classical potential theory. In the \emph{hyper-singular case}, when this energy integral is $+\infty$ for every positive measure supported on $A$, tools from geometric measure theory yield the leading-term behaviour for a large family of sets $A$. Much less is known about higher-order terms; see \cite{BrHaSa2012}. 
Recently, B{\'e}termin~\cite{Be2014arXiv} found a surprising connection between the problem of minimising a planar “Coulombian renormalized energy” (introduced by Sandier and Serfaty~\cite{SaSe2012}, also see~\cite{Se2014}) and the discrete logarithmic energy problem on the unit sphere in~$\mathbb{R}^3$ that yields the existence of the $N$-term (third term) in the asymptotics of the minimal logarithmic energy on the sphere in~$\mathbb{R}^3$. 
From this perspective it seems to be amazing that complete asymptotic expansions of optimal energy can be obtained at all. In \cite{BrHaSa2009}, we derived such an expansion for point sets on the \emph{unit circle~$\mathbb{S}$} provided with the Euclidean metric of the ambient space.\footnote{The analogous result for geodesic metric on $\mathbb{S}$ is considered in~\cite{BrHaSa2012b}. The author~\cite{Br2014manuscript} has also obtained complete asymptotic expansion of the logarithmic potential energy characterizing zeros of classical orthogonal polynomials (Chebysheff, Gegenbauer, Jacobi, Laguerre, and Hermite).}
%
Since the $N$th roots of unity (or any rotation thereof) are universally optimal point distributions on $\mathbb{S}$ due to a general result of Cohn and Kumar \cite[Theorem 1.2]{CoKu2007} and thus maximize (if $-2<s<0$) and minimize \linebreak (if $s>0$) the energy functional~\eqref{eq:energy.functional} over all $N$-point sets on $\mathbb{S}$, the Riesz $s$-energy ($s \neq 0$ real) of the $N$-th roots of unity can be computed by means of
\begin{equation} \label{LsDef}
\mathcal{L}_s(N) = 2^{-s} N \sum_{k=1}^{N-1} \left( \sin \frac{\pi k}{N} \right)^{-s}, \qquad N\ge 2, \, s\in \Cset, \, s\neq 0.
\end{equation}
We remark that for $s<-2$ optimal configurations on the unit circle with $2K$ points consist of $K$ points at each end of a diameter of $\mathbb{S}$. (This follows from results of G. Bj{\"o}rck~\cite{Bj1956}.)

The asymptotic expansion of $\mathcal{L}_s(N)$ for the general case $s \in \C$ with $s \neq 0, 1, 3, 5, \dots$ has the following form (see \cite[Theorem~1.1]{BrHaSa2009}): for fixed $p = 1, 2, 3, \dots$, 
\begin{equation} \label{eq:circle.asympt}
\begin{split}
\mathcal{L}_s(N) 
&= W_s \, N^2 + \frac{2\zetafcn(s)}{(2\pi)^s} \, N^{1+s} + 
\sum_{n=1}^p \alpha_n(s) \frac{2\zetafcn(s-2n)}{(2\pi)^s} \, N^{1+s-2n} + \mathcal{O}_{s,p}(N^{-1+\Re s-2p}) 
\end{split}
\end{equation}
as $N \to \infty$.
The coefficients $\alpha_n(s)$, $n\geq0$, satisfy the generating function relation
\begin{equation*} \label{sinc.power.0}
\left( \frac{\sin \pi z}{\pi z} \right)^{-s} = 
\sum_{n=0}^\infty \alpha_n(s) \, z^{2n}, \quad |z|<1, \ s\in \mathbb{C}.
\end{equation*}
%
%
%
In the potential-theoretic case $0 < s < 1$, the constant $W_s$ is the Riesz $s$-energy of $\mathbb{S}$; i.e.,  
\begin{equation} \label{eq:W.s}
W_s = W_s(\mathbb{S}) = 2^{-s} \frac{\gammafcn((1-s)/2)}{\sqrt{\pi}\gammafcn(1-s/2)}
\end{equation}
is the value of the energy integral for the uniform measure on $\mathbb{S}$ that uniquely minimizes \eqref{eq:energy.integral} amongst all Borel probability measures supported on $\mathbb{S}$. 
In general, $W_s$ will be identified with the right-hand side of~\eqref{eq:W.s}. The use of analytic continuation in the complex variable $s$ provides a unifying approach to the energy asymptotics and explains why certain phenomena occur. 
This \emph{principle of analytic continuation} is a fundamental feature which is assumed to hold also on higher-dimensional spheres.
Here it explains how the interplay between the infinitely many poles of $W_s$ at $s = 1, 3, 5, \dots$ and the poles of the shifted Riemann zeta functions $\zetafcn(s-2n)$ in \eqref{eq:circle.asympt} gives rise to a blow up of constants whenever one of the terms $N^{1+s-2n}$ of the asymptotic scale approaches the $N^2$term and why a logarithmic term appears in the exceptional cases $s = 1, 3, 5, \dots$. 
%


On the other hand, the Riemann zeta function is zero at negative even integers. This implies that the sum in \eqref{eq:circle.asympt} vanishes for all $N \geq 2$ if $s$ is a negative even integer or can have at most $m$ terms if $s = 2m$ is a positive even integer. Thus, the asymptotic expansion gives explicit formulas for $\mathcal{L}_s(N)$ that are correct up to a remainder term of order $\mathcal{O}_{s,p}(N^{-1-s-2p})$ for sufficiently large~$N$ when $s$ is an even integer. The purpose of this paper is to derive explicit expressions for $\mathcal{L}_{2m}(N)$, $m \neq 0$ an integer, that are valid for all $N \geq 2$.


\emph{Outline.} Sections~\ref{sec:m.positive} and \ref{sec:m.negative} discuss the cases for negative and positive even $s$, respectively. The singular Riesz $s$-potential in Section~\ref{sec:m.negative} is approximated by means of a family of continuous kernels. The explicit computation of the coefficients of the powers of $N$ involves the use of Stirling numbers of the first kine, Eulerian numbers, and evaluation of certain Bell polynomials. The proofs of the results in  Section~\ref{sec:m.negative} are given in Section~\ref{sec:proofs}. Auxiliary results are collected in Appendix~\ref{sec:appendix}.



\section{The case of negative even integers $s$}
\label{sec:m.positive}

In \cite[Remark~1, Eq.~(1.19)]{BrHaSa2009} it was observed that $\mathcal{L}_{s}(N) = W_s \, N^2$ for $s=-2,-4,-6,\dots$, 
where $W_s$ is given in \eqref{eq:W.s}. This relation holds, indeed, for sufficiently large $N$. The precise result is

\begin{thm}
Let $m$ be a positive integer and set $s = -2m$. Then
\begin{equation*}
\mathcal{L}_s(N) = W_s \, N^2 + 2 N^2 \sum_{\substack{k=1 \\ N | k}}^m (-1)^k \binom{2m}{m-k}, \qquad N = 2, 3, 4, \dots.
\end{equation*}
The sum consists of terms for which the integer $N$ divides the integer $k$ and vanishes if $N > m$.
\end{thm}



\begin{proof}
Using the identity (\cite[Eq.~I.1.9]{PrBrMa1986I})
\begin{equation*}
\left( \sin \phi \right)^{2m} = 2^{1-2m} \sum_{k=0}^{m-1} (-1)^{m-k} \binom{2m}{k} \cos[ 2(m-k) \phi] + 2^{-2m} \binom{2m}{m}
\end{equation*}
in \eqref{LsDef}, one has
\begin{equation*}
\mathcal{L}_{-2m}(N) 
= 2 N \sum_{\ell=1}^m (-1)^\ell \binom{2m}{m-\ell} \left[ \sum_{k=0}^{N-1} \cos \frac{2\pi k \ell}{N} \right] + \binom{2m}{m} \, N^2.
\end{equation*}
The trigonometric sums can be summed up (\cite[Equations~4.4.4.6--8]{PrBrMa1986I}),
\begin{equation} \label{eq:cos.identity}
\sum_{k=0}^{N-1} \cos ( 2 \pi k \ell / N ) = 
\begin{cases} 
N & \text{if $N$ divides $\ell$ (i.e., $N | \ell$),} \\
0 & \text{otherwise.}
\end{cases}
\end{equation}
Using the duplication formula for the gamma function, the constant $W_s$ can be written as $\binom{2m}{m}$.
%
The result follows.
\end{proof}

\section{The case of positive even integers $s$}
\label{sec:m.negative}

In \cite[Remark~1, Eq.~(1.20)]{BrHaSa2009} it was observed that 
\begin{equation} \label{exp2xct}
\mathcal{L}_s(N) = \sum_{n=0}^m \alpha_n(s) \frac{2\zetafcn(s-2n)}{(2\pi)^s} \, N^{1+s-2n}, \qquad s=2, 4, 6, \dots.
\end{equation}
Indeed, $W_s$ of \eqref{eq:W.s} vanishes if $s$ is a positive even integer and the shifted Riemann zeta functions $\zetafcn(s-2n)$ are zero if $s-2n$ is an even negative integer. 
Nevertheless, \eqref{exp2xct} is an asymptotic formula which holds up to a  
term of order $\mathcal{O}_{s,p}(N^{-1+s-2p})$, where $p$ is any integer $\geq s/2$. 
In fact, direct computation yields that the formula \eqref{exp2xct} holds for all $N \geq 2$ (see Theorem~\ref{thm:even.positive.s} below). 
%

%
%
The distance between a point $\PT{x}$ on the unit circle ($|\PT{x}| = 1$) and a point $\PT{y}$ inside the unit circle ($r=|\PT{y}| < 1$) is given by 
$\left| \PT{x} - \PT{y} \right|^2 = 1 - 2 r \langle \PT{x}, \PT{y} \rangle + r^2 = \left( 1 - r \right)^2 + 4 r \left[ \sin( \phi / 2 ) \right]^2$, $\cos \phi = \langle \PT{x}, \PT{y} \rangle$.
This motivates the following modification of the the $s$-energy of the $N$-th roots of unity:
\begin{equation} \label{eq:callM}
\mathcal{M}_{s}( N; r ) \DEF N \sum_{k = 1}^{N-1} \left( 1 - 2 r \cos \frac{2 \pi k}{N} + r^2 \right)^{-s/2}, \qquad 0 \leq r < 1.
\end{equation}
Chu~\cite{Chu2003} studied similar sums which satisfy a recurrence relation,
\begin{equation*}
f_m(z) = \sum_{k=0}^{n-1} \frac{1 - z^2}{\left[ 1 - 2 z \cos( 2 \pi k / n ) + z^2 \right]^m}, \qquad f_{m+1}(z) = \frac{f_m(z)}{1-z^2} + \frac{z}{m} \frac{\dd}{\dd z} \left\{ \frac{f_m(z)}{1-z^2} \right\}.
\end{equation*}
However, his main interest was in the case when $z = \sqrt{-1}$.

In the following we make use of the {\em Gauss hypergeometric function}
\begin{equation} \label{eq:Gauss.hypergeometric}
\Hypergeom{2}{1}{a,b}{c}{z} = \sum_{n=0}^{\infty} \frac{\Pochhsymb{a}{n} \Pochhsymb{b}{n}}{\Pochhsymb{c}{n}} \; \frac{z^n}{n!}, \qquad |z| < 1,
\end{equation}
where $\Pochhsymb{a}{n}$ denotes the Pochhammer symbol defined by $\Pochhsymb{a}{0} \DEF 1$ and $\Pochhsymb{a}{n+1} = \Pochhsymb{a}{n} ( a + n )$. One has the relation $\Pochhsymb{a}{n} = \gammafcn(n + a ) / \gammafcn(a)$ provided the right-hand side is well-defined.

\begin{lem} \label{lem:mod.energy}
Let $s > 0$. Then for integers $N \geq 2$,
\begin{equation*}
\mathcal{M}_{s}( N; r ) = N^2 \, G_0(s; r) - N \left( 1 - r \right)^{-s} + 2 N^2 \sum_{\nu=1}^\infty G_{\nu N}(s; r), \qquad 0 \leq r < 1,
\end{equation*}
where the series is uniformly convergent with respect to $r$ on compact subsets of $[0,1)$ and
\begin{equation*}
G_n(s; r) = \left( 1 - r^2 \right)^{1-s} \frac{\Pochhsymb{s/2}{n}}{n!} r^{n} \Hypergeom{2}{1}{1-s/2,n + 1 - s/2}{n+1}{r^2}.
\end{equation*}
\end{lem}

For $s = 2m$ a positive even integer, one can write (cf. \cite[Eq.~15.8.7]{Olver:2010:NHMF})
\begin{equation} \label{eq:G.n.2m.r}
G_n(2m; r) = \left( 1 - r^2 \right)^{1-2m} \frac{2^{2m-2} \gammafcn(m-1/2)}{\sqrt{\pi}\, \gammafcn(m)} r^{n} \Hypergeom{2}{1}{1-m,n + 1 - m}{2-2m}{1-r^2}.
\end{equation}
The hypergeometric function above reduces to a polynomial of degree $m-1$.

For $s=2$ (i.e., $m=1$), one simply has $G_n(2;r) = r^n / (1 - r^2)$. Hence,
\begin{align*}
\mathcal{M}_{2}( N; r ) 
&= \frac{N^2}{1 - r^2} - \frac{N}{\left( 1 - r \right)^{2}} + \frac{2 N^2}{1 - r^2} \sum_{\nu=1}^\infty r^{\nu N} 
= \frac{1 + r^N}{\left( 1 - r^2 \right) \left(1- r^N\right)} \, N^2  - \frac{N}{\left( 1 - r \right)^{2}}.
\end{align*}
With the help of Mathematica, 
\begin{equation*}
\mathcal{L}_2(N) = \lim_{r\to1^-} \mathcal{M}_{2}( N; r ) = \frac{N^3}{12} - \frac{N}{12} \qquad \text{and by \eqref{exp2xct}:} \qquad \mathcal{L}_2(N) = \frac{2}{4 \pi^2} \left( \frac{\pi^2}{6} N^3 - \frac{\pi^2}{6} N \right).
\end{equation*}
It follows that \eqref{exp2xct} is correct for $s = 2$ and all $N \geq 1$. 
In the general result we make use of {\em Stirling numbers of the first kind $\Stirlings(n,k)$} defined by the expansion (\cite[Eq.~26.8.7]{Olver:2010:NHMF})
\begin{equation} \label{eq:Stirling.s}
\sum_{k=0}^n \Stirlings(n,k) x^k = x \left( x - 1 \right) \cdots \left( x - n + 1 \right) = \Pochhsymb{x+1-n}{n},
\end{equation}
whereas the {\em Eulerian numbers $\Eulerian{n}{k}$} are given by (cf. \cite[Equations~26.14.5 and 26.14.6]{Olver:2010:NHMF})
\begin{equation} \label{eq:Eulerian}
\sum_{k=0}^{n-1} \Eulerian{n}{k} \binom{x+k}{n} = x^n, \qquad \Eulerian{n}{k} = \sum_{j=0}^k (-1)^j \binom{n+1}{j} \left( k + 1 - j \right)^n \quad n \geq 1.
\end{equation}
There holds $\Eulerian{0}{0} = \Eulerian{n}{0} = 1$ and $\Eulerian{n}{n-1} = 1$. By convention, $\Eulerian{n}{k} = 0$ for $k \geq n$ and $(n,k)\neq(0,0)$.
\begin{thm} \label{thm:2}
Let $m$ and $N$ be positive integers. Then
\begin{equation*}
\begin{split}
\mathcal{M}_{2m}(N; r) 
&= \frac{N^2}{\gammafcn(m)} \sum_{k=0}^{m-1} \frac{\gammafcn(2m-k-1)}{\gammafcn(m-k) k!} \left[ - \Pochhsymb{1-m}{k} + 2 \sum_{q=0}^{k} \frac{g(k,q;N,1-m)}{\left( 1 - r^N \right)^{q+1}} \right] \\
&\phantom{=\pm}\times \left( 1 - r^2 \right)^{k+1-2m} - N \left( 1 - r \right)^{-2m}, \qquad 0 \leq r < 1,
\end{split}
\end{equation*}
where
\begin{equation*}
g(k,q;a,b) = \sum_{p=0}^{k-q} (-1)^{p} s(k,p+q;b) b(p+q,p) a^{p+q}
\end{equation*}
with
\begin{equation*}
b(p+q,p) = \sum_{j=0}^q \Eulerian{p+q}{j} \binom{p+q-j}{p}, \qquad \Stirlings(n, \ell; y) = \sum_{k=\ell}^n \binom{k}{\ell} \Stirlings(n,k) \left( y+n-1 \right)^{k-\ell}.
\end{equation*}

%
\end{thm}

We are interested in the limit as $r \to 1^-$. Clearly, $\mathcal{M}_{2m}(N) \DEF \lim_{r\to1^-} \mathcal{M}_{2m}(N; r) = \mathcal{L}_{2m}(N)$.
%

\begin{thm} \label{thm:even.positive.s}
Let $m$ and $N$ be positive integers. Then
\begin{equation} \label{eq:M.2m.N.expansion}
\mathcal{M}_{2}(N) = \frac{N^3}{12} - \frac{N}{12}, \qquad \mathcal{M}_{2m}(N) = \sum_{\nu = 0}^{2m} \beta_{2m-\nu}(m) \, N^{1+2m-\nu}, \quad m \geq 2,
\end{equation}
where the coefficients $\beta_\nu(m)$ ($0 \leq \nu \leq 2m$) do not depend on $N$. 
\end{thm}

Explicit expressions for $\beta_\nu(m)$ are obtained in the proof of Theorem~\ref{thm:even.positive.s} (see Equations~\eqref{eq:beta.coeff.s}). The expansion \eqref{eq:M.2m.N.expansion} holds for every integer $N \geq2$. The expansion \eqref{exp2xct}, derived from an asymptotic result as $N \to \infty$,  represents the same quantity and holds for sufficiently large $N$. The coefficients of the powers of $N$ in either expansion do not depend on $N$. By comparing \eqref{eq:beta.coeff.s} and~\eqref{exp2xct}, one has the connection formulas
\begin{equation*}
\beta_{2n+1}(m) = 0 \quad (0 \leq n < m), \qquad \beta_{2m-2n}(m) = \frac{2 \zetafcn(2m-2n)}{(2\pi)^{2m}} \alpha_n(2m) \quad (0 \leq n \leq m).
\end{equation*}
The coefficients $\alpha_n(2m)$ can be expressed in terms of {\em generalized Bernoulli polynomials} $B_k^{(\sigma)}(x)$: 
\begin{equation} \label{eq:alpha.n.s}
\alpha_n(2m) = \frac{(-1)^n B_{2n}^{(2m)}(m)}{(2n)!} \left( 2 \pi \right)^{2n}, \quad n \geq 0.
\end{equation}
For non-negative even integers, the Riemann zeta function can be expressed in terms of Bernoulli numbers by means of $\zeta(2m) = 2^{2m-1} (-1)^{m-1} B_{2m} \pi^{2m} / (2m)!$, $m = 0,1,2, \dots$.
This leads to 
\begin{equation*}
\beta_{2m-2n}(m) = \frac{(-1)^{m-n-1} B_{2m-2n}}{(2m-2n)!} \frac{(-1)^n B_{2n}^{(2m)}(m)}{(2n)!}, \qquad 0 \leq n \leq m.
\end{equation*}

For $n=0$, one gets the following identity.
\begin{cor} \label{cor:4}
For $m = 1, 2, \dots$, one has
\begin{equation*}
\frac{(-1)^{m-1} B_{2m}}{(2m)!} = \sum_{\nu=0}^{m-1} \sum_{p=0}^{\nu} (-1)^p \frac{2^{2+\nu-2m} \gammafcn(2m-1-\nu)}{\gammafcn(m) \nu! \gammafcn(m-\nu)} b(\nu,p) G(2m-p,2m-p,\nu-p+1),
\end{equation*} 
where the numbers $b(\nu,p)$ and $G(n,n,q)$ are given by
\begin{equation*}
b(k,p) = \sum_{j=0}^{k-p} \Eulerian{k}{j} \binom{k-j}{p}, \qquad \sum_{n=1}^\infty (-1)^n G(n,n,q) \, x^n = \left( \frac{e^x-1}{x} \right)^{-q}-1.
\end{equation*}
\end{cor}


\section{Proofs}
\label{sec:proofs}

The Gegenbauer polynomials $C_n^{(\lambda)}$ of degree $n$ with index $\alpha>0$ can be defined by means of 
\begin{equation} \label{eq:Gegenbauer.C.generating}
\sum_{n=0}^\infty z^n \, C_n^{(\lambda)}(\cos \phi) = \frac{1}{\left(1 - 2 z \cos \phi + z^2 \right)^{\lambda}}, \qquad | z | < 1, \lambda > 0.
\end{equation}
They can be represented through trigonometric functions (see \cite[Section~22]{AbSt1992}), 
\begin{equation} \label{eq:Gegenbauer.C.trigonometric}
C_n^{(\lambda)}(\cos \phi) = \sum_{\ell=0}^{n} \frac{\Pochhsymb{\lambda}{\ell} \Pochhsymb{\lambda}{n - \ell}}{\ell! (n - \ell)!} \cos[ ( n - 2 \ell ) \phi ], \qquad \lambda \neq 0.
\end{equation}
The summation formula~\eqref{eq:cos.identity} yields the following auxiliary result.
%

\begin{lem} \label{lem:aux.1}
Let $s > 0$. Then
\begin{equation*}
\mathcal{A}_n(s;N) \DEF N \sum_{k = 1}^{N-1} C_n^{(s/2)}(\cos \frac{2\pi k}{N}) = N^2 \sum_{\substack{\ell=0 \\ N | ( n - 2 \ell )}}^{n} \frac{\Pochhsymb{s/2}{\ell} \Pochhsymb{s/2}{n - \ell}}{\ell! (n - \ell)!} - N \, C_n^{(s/2)}(1).
\end{equation*}
\end{lem}


\begin{proof}[Proof of Lemma~\ref{lem:mod.energy}]
Using \eqref{eq:Gegenbauer.C.generating} in \eqref{eq:callM} and Lemma~\ref{lem:aux.1}, we obtain
\begin{align*}
\mathcal{M}_{s}( N; r ) 
&= \sum_{n=0}^\infty r^n \, N \sum_{k = 1}^{N-1} C_n^{(s/2)}(\cos \frac{2\pi k}{N}) = N^2  \, \sum_{n=0}^\infty r^n \sum_{\substack{\ell=0 \\ N | ( n - 2 \ell )}}^{n} \frac{\Pochhsymb{s/2}{\ell} \Pochhsymb{s/2}{n - \ell}}{\ell! (n - \ell)!} - N \sum_{n=0}^\infty r^n \, C_n^{(s/2)}(1).
\end{align*}
The last series has the closed form $(1-r)^{-s}$ (cf. \eqref{eq:Gegenbauer.C.generating}). One has
\begin{align*}
\sum_{\nu=0}^\infty r^{2\nu} \sum_{\substack{\ell=0 \\ N | ( 2\nu - 2 \ell )}}^{2\nu} \frac{\Pochhsymb{s/2}{\ell} \Pochhsymb{s/2}{2\nu - \ell}}{\ell! (2\nu - \ell)!} &= \sum_{\nu=0}^\infty \frac{\Pochhsymb{s/2}{\nu} \Pochhsymb{s/2}{\nu}}{\nu! \nu!} r^{2\nu} + 2 \sum_{\nu=0}^\infty r^{2\nu} \sum_{\substack{\ell = 1 \\ N | (2\ell)}}^{\nu} \frac{\Pochhsymb{s/2}{\nu-\ell} \Pochhsymb{s/2}{\nu + \ell}}{(\nu-\ell)! (\nu + \ell)!} \\
&= \Hypergeom{2}{1}{s/2,s/2}{1}{r^2} + 2 \sum_{\substack{\ell = 1 \\ N | (2\ell)}}^{\infty} \frac{\Pochhsymb{s/2}{2\ell}}{(2\ell)!} r^{2\ell} \Hypergeom{2}{1}{s/2,2\ell + s/2}{2\ell+1}{r^2}.
\end{align*}
In the last step \eqref{eq:Gauss.hypergeometric} was used. A similar formula (without the hypergeometric function part and $2\ell \mapsto 2\ell+1$) holds for the sum over odd powers of $r$. 
%
Putting everything together, we obtain
\begin{equation*}
\mathcal{M}_{s}( N; r ) = N^2 \, G_0(s; r) - \frac{N}{\left( 1 - r \right)^{s}} + 2 N^2 \sum_{\substack{n=1 \\ N | n}}^\infty G_n(s; r), \quad G_n(s; r) \DEF \frac{\Pochhsymb{s/2}{n}}{n!} r^{n} \Hypergeom{2}{1}{s/2,n + s/2}{n+1}{r^2}.
\end{equation*}
The series $\sum_{\nu=1}^\infty G_{\nu N}(s; r)$ converges uniformly with respect to $r$ on compact subsets of $[0,1)$, which can be seen from the integral representation (cf. \cite[Eq.~15.6.1]{Olver:2010:NHMF}) and estimate
\begin{align*}
G_n(s; r) 
&= r^n \frac{1}{[\gammafcn(s/2)]^2} \frac{\gammafcn(n+s/2)}{\gammafcn(n+1-s/2)} \int_0^1 \frac{t^{s/2-1} \left( 1 - t \right)^{n-s/2}}{\left( 1 - r^2 t \right)^{n+s/2}} \dd t \\
&\leq  c_s(r) \frac{\gammafcn(n+s/2)}{\gammafcn(n+1-s/2)} \, r^{n}, \qquad c_s(r) = \frac{1}{[\gammafcn(s/2)]^2} \int_0^1 \frac{t^{s/2-1}}{\left( 1 - r^2 t \right)^{s}} \dd t,
\end{align*}
valid for $n \geq n_0 > s / 2$, where $c_s(r)$ can be uniformly bounded on compact sets in $[0,1)$. Hence
\begin{equation*}
0 < \sum_{\nu=1}^\infty G_{\nu N}(s; r) \leq C + \sum_{n=n_0}^\infty G_{n}(s; r) \leq C + c_s(r) \sum_{n=n_0}^\infty \frac{\gammafcn(n+s/2)}{\gammafcn(n+1-s/2)} \, r^{n} \quad \text{for some $C>0$} 
\end{equation*}
and the right-most series above converges uniformly on compact sets in $[0,1)$ for $n_0 > s/2 > 0$.

Application of the last linear transformation in \cite[Equations~15.8.1]{Olver:2010:NHMF} yields
\begin{align*}
G_n(s; r) = \left( 1 - r^2 \right)^{1-s} \frac{\Pochhsymb{s/2}{n}}{n!} r^{n} \Hypergeom{2}{1}{1-s/2,n + 1 - s/2}{n+1}{r^2}.
\end{align*}
For $s$ a positive even integer, that is $s = 2m$, we can write (cf. \cite[Eq.~15.8.7]{Olver:2010:NHMF})
\begin{equation*}
G_n(2m; r) = \left( 1 - r^2 \right)^{1-2m} \frac{\Pochhsymb{m}{n}}{n!} \frac{\Pochhsymb{m}{m-1}}{\Pochhsymb{n+1}{m-1}} r^{n} \Hypergeom{2}{1}{1-m,n + 1 - m}{2-2m}{1-r^2},
\end{equation*}
where the ratios can be simplified further. This shows \eqref{eq:G.n.2m.r}.
\end{proof}

\begin{proof}[Proof of Theorem~\ref{thm:2}]
The series expansion of the hypergeometric polynomial in \eqref{eq:G.n.2m.r} yields
\begin{equation*}
G_n(2m; r) = \frac{2^{2m-2} \gammafcn(m-1/2)}{\sqrt{\pi}\, \gammafcn(m)} r^{n} \sum_{k=0}^{m-1} \frac{\Pochhsymb{1-m}{k} \Pochhsymb{n+1-m}{k}}{\Pochhsymb{2-2m}{k} k!} \left( 1 - r^2 \right)^{k+1-2m}.
\end{equation*}
Since
\begin{equation*}
\frac{\Pochhsymb{1-m}{k}}{\Pochhsymb{2-2m}{k}} = \binom{m-1}{k} \Big/ \binom{2m-2}{k} = \frac{\gammafcn(m) \gammafcn(2m-1-k)}{\gammafcn(m-k) \gammafcn(2m-1)}
\end{equation*}
and $\gammafcn(2m-1) = 2^{2m-2} \gammafcn(m-1/2) \gammafcn(m) / \sqrt{\pi}$ (duplication formula for $\gammafcn$), we obtain
\begin{equation*}
G_n(2m; r) = \frac{1}{\gammafcn(m)} r^{n} \sum_{k=0}^{m-1} \frac{\gammafcn(2m-1-k)}{\gammafcn(m-k) k!} \frac{\Pochhsymb{n+1-m}{k}}{k!} \left( 1 - r^2 \right)^{k+1-2m}.
\end{equation*}
By Lemma~\ref{lem:mod.energy} and the last representation we get
\begin{equation*}
\begin{split}
\mathcal{M}_{2m}(N; r) 
&= \frac{N^2}{\gammafcn(m)} \sum_{k=0}^{m-1} \frac{\gammafcn(2m-k-1)}{\gammafcn(m-k) k!} \frac{ \sum_{\nu = - \infty}^\infty \Pochhsymb{| \nu | N + 1 - m}{k} r^{| \nu | N} }{ \left( 1 - r^2 \right)^{2m-k-1} } - N \left( 1 - r \right)^{-2m}.
\end{split}
\end{equation*}
By Lemma~\ref{lem:aux.4}, the infinite series reduces to a rational function in $r^N$. The result follows. 
\end{proof}

\subsection*{Proof of Theorem~\ref{thm:even.positive.s}}

For the proof of Theorem~\ref{thm:even.positive.s} we need some preparations. We define 
\begin{equation*}
f_q(r) \DEF \left( \frac{1-r^N}{1-r} \right)^{-q}, \qquad h(r) \DEF \frac{1-r^N}{1-r} = \sum_{\ell = 0}^{N-1} r^\ell
\end{equation*}
and get a series expansion of $f_q(r)$ at $r = 1$. We use Fa\`{a} di Bruno's differentiation formula
\begin{align*}
\left\{ f(g(x)) \right\}^{(n)} 
&= \sum \frac{n!}{k_1! \cdots k_{n}!} f^{(k)}(g(x)) \prod_{\nu=1}^n \left( \frac{g^{(\nu)}(x)}{\nu!} \right)^{k_\nu} = \sum_{k=1}^n f^{(k)}(g(x)) \BellB_{n,k}(g^\prime(x), g^{\prime\prime}(x), \dots),
\end{align*}
where $k = k_1 + \cdots + k_{n}$ in the first sum and this sum is extended over all partition of $n$, that is integers $k_1, \dots, k_{n} \geq 0$ such that $k_1 + 2 k_2 + \cdots + n k_{n} = n$. The polynomials $\BellB_{n,k}(x_1, x_2, \dots)$ in the second sum are the {\em (partial) Bell polynomials}, explicitly given by
\begin{equation} \label{eq:Bell.polynomial}
\BellB_{n,k}(x_1, x_2, \dots) = \sum_{\substack{k_1 + k_2 + \cdots = k \\ k_1 + 2 k_2 + \cdots = n}} \frac{n!}{k_1! k_2! \cdots} \left( \frac{x_1}{1!} \right)^{k_1} \left( \frac{x_2}{2!} \right)^{k_2} \cdots.
\end{equation}
They satisfy the generating function relation (see, e.g., \cite{Mi2010})
\begin{equation} \label{eq:Bell.polynomial.generating.rel}
\sum_{n=k}^{\infty} \BellB_{n,k}(x_1, x_2, x_3, \dots) \frac{t^n}{n!} = \frac{1}{k!} \Big( \sum_{m=1}^\infty x_m \frac{t^m}{m!} \Big)^k.
\end{equation}
(In the polynomial $\BellB_{n,k}$ are only the variables $x_1, x_2, \dots, x_{n-k+1}$ active.) We record that
\begin{subequations} \label{eq:BellB.id.1}
\begin{align}
\BellB_{n,k}(\alpha x_1, \alpha x_2, \dots) &= \alpha^k \BellB_{n,k}(x_1, x_2, \dots), \qquad \alpha \neq 0, \\
\BellB_{n,k}(\alpha^1 x_1, \alpha^2 x_2, \dots) &= \alpha^n \BellB_{n,k}(x_1, x_2, \dots), \qquad \alpha \neq 0.
\end{align}
\end{subequations}

\begin{lem} \label{lem:aux.5}
Let $n$ be an integer $\geq 0$. Then
\begin{equation*}
h^{(n)}(1) = \frac{N \left( N - 1 \right) \cdots \left( N - n \right)}{n+1} = n! \binom{N}{n+1}. 
\end{equation*}
\end{lem}

\begin{proof}
With the convention $\sum_{\ell = n}^M a_n = 0$ for $M < n$, one has
\begin{equation*}
h^{(n)}(r) = \left\{ \frac{1-r^N}{1-r} \right\}^{(n)} = \Big\{ \sum_{\ell = 0}^{N-1} r^\ell \Big\}^{(n)} = \sum_{\ell = 0}^{N-1} \left\{ r^\ell \right\}^{(n)} = \sum_{\ell = n}^{N-1} \ell \left( \ell - 1 \right) \cdots \left( \ell - n + 1 \right) r^{\ell-n}.
\end{equation*}
If $r = 1$, then
\begin{align*}
h^{(n)}(1) 
= \sum_{\ell = 1}^{N-n} \left( \ell + n - 1 \right) \cdots \left( \ell + 1 \right) \ell = \sum_{\ell = 1}^{N-n} \Pochhsymb{\ell}{n} = \frac{N \left( N - 1 \right) \cdots \left( N - n \right)}{n+1}.
\end{align*}
The last step follows by complete induction in $N$.
\end{proof}

\begin{lem} \label{lem:aux.6}
Let $q$ be a positive real number and $N$ a positive integer. Then for $m = 1, 2, \dots$
\begin{equation*}
f_q(r) = N^{-q} + \sum_{n=1}^{m} \frac{f_q^{(n)}(1)}{n!} \left( r - 1 \right)^n + \mathcal{R}_m(r), \qquad 0 < r < 1,
\end{equation*}
where the coefficients $f_q^{(n)}(1)$ can be written as
\begin{equation*}
f_q^{(n)}(1) = \sum_{k=1}^n (-1)^k \Pochhsymb{q}{k} N^{-q-k} \, \BellB_{n,k}(1! \binom{N}{2}, 2! \binom{N}{3}, 3! \binom{N}{4}, \dots)
\end{equation*}
and the remainder can be estimated by 
\begin{equation*}
\left| \mathcal{R}_m(r) \right| \leq \frac{\left( 1 - r \right)^{m+1}}{(m+1)!} \sum_{k=1}^{m+1} \Pochhsymb{q}{k} f_{q+k}(r) \, \BellB_{m+1,k}(1! \binom{N}{2}, 2! \binom{N}{3}, 3! \binom{N}{4}, \dots).
\end{equation*}
\end{lem}

\begin{proof}
Let $m \geq 1$. By Taylor's theorem
\begin{equation*}
f_q(r) = \left( \frac{1-r^N}{1-r} \right)^{-q} = \sum_{n=0}^{m} \frac{f_q^{(n)}(1)}{n!} \left( r - 1 \right)^n + \mathcal{R}_m(r), \qquad \mathcal{R}_m(r) = \frac{f_q^{(m+1)}(\rho)}{(m+1)!} \left( r - 1 \right)^{m+1}
\end{equation*}
for some $\rho$ with $r < \rho < 1$.  By Fa\`{a} di Bruno's differentiation formula
\begin{align*}
f_q^{(n)}(r) 
&= \left\{ \left( \frac{1-r^N}{1-r} \right)^{-q} \right\}^{(n)} = \sum_{k=1}^n \left\{(\cdot)^{-q}\right\}^{(k)}(h(r)) \BellB_{n,k}(h^\prime(r), h^{\prime\prime}(r), \dots) \\
&= \sum_{k=1}^n (-1)^k \Pochhsymb{q}{k} \left( h(r) \right)^{-q-k} \, \BellB_{n,k}(h^\prime(r), h^{\prime\prime}(r), \dots) .
\end{align*}

The representation of $f_q^{(n)}(1)$ follows from the fact $h(1) = N$ and Lemma~\ref{lem:aux.5}.
%
%

The remainder $\mathcal{R}_m(r)$ is estimated next. Observe that the function $h^{(\nu)}(r)$ is positive and strictly monotonically increasing on $(0,1)$ for each $\nu \geq 0$. Thus, by the triangle inequality,
\begin{equation*}
\left| f_q^{(m+1)}(\rho) \right| \leq \sum_{k=1}^{m+1} \Pochhsymb{q}{k} \left( h(\rho) \right)^{-q-k} \left| \BellB_{m+1,k}(h^\prime(r), h^{\prime\prime}(r), \dots) \right|
\end{equation*}
and by \eqref{eq:Bell.polynomial}, the Bell polynomial is non-negative. 
The estimate follows.
\end{proof}

Let $\binom{a_1+\dots+a_k}{a_1,\dots,a_k}$ denote the multinomial coefficient.

The specific Bell polynomials appearing in Lemma~\ref{lem:aux.6} can be represented as follows.

\begin{lem} \label{lem:aux.7}
Let $n$, $k$ and $N$ be positive integers. Then
\begin{equation*}
\begin{split} \label{eq:Bell.Poly.Identity.A}
\BellB_{n,k}(1! \binom{N}{2}, 2! \binom{N}{3}, 3! \binom{N}{4}, \dots) 
&= \frac{n!}{k!} \mathop{\sum_{n_1=1}^{N-1} \cdots \sum_{n_k=1}^{N-1}}_{n_1+\cdots+n_k=n} \binom{N}{n_1+1} \cdots \binom{N}{n_k+1} \\
&= \frac{n!}{k! (n+k)!} \sum_{\ell=k}^{k+n} (-1)^{n-\ell+k} H_\ell(n,k) \, N^\ell,
\end{split}
\end{equation*}
where the coefficients
\begin{equation*}
H_\ell(n,k) = \mathop{\sum_{n_1=1}^{n} \cdots \sum_{n_k=1}^{n}}_{n_1+\cdots+n_k=n} \binom{n+k}{n_1+1,n_2+1,\dots,n_k+1} \mathop{\sum_{\ell_1=1}^{n_1+1} \cdots \sum_{\ell_k=1}^{n_k+1}}_{\ell_1+\cdots+\ell_k=\ell} \left| \Stirlings(n_1+1, \ell_1) \right| \cdots \left| \Stirlings(n_k+1, \ell_k) \right|.
\end{equation*}
vanish for $\ell = 0, 1, \dots, k - 1$ and do not depend on $N$. 
\end{lem}

\begin{proof}
By the generating function relation for partial Bell polynomials \eqref{eq:Bell.polynomial.generating.rel}
\begin{equation*}
\sum_{n = k}^\infty \BellB_{n,k}(1! \binom{N}{2}, 2! \binom{N}{3}, 3! \binom{N}{4}, \dots) \frac{t^n}{n!} = \frac{1}{k!} \left( \sum_{m=1}^\infty \binom{N}{m+1} t^m \right)^k.
\end{equation*}
The right-hand side above evaluates as 
\begin{align*}
\frac{1}{k!} \left[ \frac{\left( 1 + t \right)^N - 1 - N \, t}{t} \right]^k 
&= \sum_{\nu=k}^{k(N-1)} \Bigg\{ \frac{\nu!}{k!} \mathop{\sum_{\nu_1=1}^{N-1} \cdots \sum_{\nu_k=1}^{N-1}}_{\nu_1 + \cdots + \nu_k = \nu} \binom{N}{\nu_1+1} \cdots \binom{N}{\nu_k+1} \Bigg\} \frac{t^\nu}{\nu!}.
\end{align*}
Comparison of coefficients yields the first identity \eqref{eq:Bell.Poly.Identity.A}. 
Let $\fallingFactorial{x}{n}$ denote the {\em falling factorial} $\fallingFactorial{x}{n} = x ( x - 1 ) \cdots ( x - n + 1)$. Taking into account that $\binom{N}{\nu+1} = \fallingFactorial{N}{\nu+1} / (\nu+1)!$
vanishes if $\nu+1 > N$ for the positive integer $N$, we may write
\begin{equation*}
\BellB_{n,k}(1! \binom{N}{2}, 2! \binom{N}{3}, 3! \binom{N}{4}, \dots) = \frac{n!}{k!} \mathop{\sum_{n_1=1}^{n} \cdots \sum_{n_k=1}^{n}}_{n_1+\cdots+n_k=n} \frac{\fallingFactorial{N}{n_1+1}}{(n_1+1)!} \cdots \frac{\fallingFactorial{N}{n_k+1}}{(n_k+1)!}.
\end{equation*}
By the definition of Stirling numbers of the first kind and using that $n_1+\cdots + n_k = n$, one has
\begin{align*}
\frac{\fallingFactorial{N}{n_1+1}}{(n_1+1)!} \cdots \frac{\fallingFactorial{N}{n_k+1}}{(n_k+1)!} 
&= \prod_{\nu=1}^k \left[ \frac{1}{(n_\nu+1)!} \sum_{\ell=0}^{n_\nu+1} s(n_\nu+1, \ell) \, N^\ell \right] \\
&= \sum_{\ell=0}^{k+n} \Bigg\{ \mathop{\sum_{\ell_1=0}^{n_1+1} \cdots \sum_{\ell_k=0}^{n_k+1}}_{\ell_1 + \cdots + \ell_k = \ell} \frac{s(n_1+1, \ell_1) \cdots s(n_k+1, \ell_k)}{(n_1+1)! \cdots (n_k+1)!} \Bigg\} N^\ell.
\end{align*}
When expanding the left-hand side above in powers of $N$, we infer that the lowest power appearing at the right-hand side above is $N^k$. That is
\begin{equation*}
\mathop{\sum_{\ell_1=0}^{n_1+1} \cdots \sum_{\ell_k=0}^{n_k+1}}_{\ell_1 + \cdots + \ell_k = \ell} \frac{s(n_1+1, \ell_1) \cdots s(n_k+1, \ell_k)}{(n_1+1)! \cdots (n_k+1)!} = 0 \qquad \text{for $\ell = 0, 1, \dots, k - 1$.}
\end{equation*}
Putting everything together, we get the identity
\begin{equation*}
\BellB_{n,k}(1! \binom{N}{2}, 2! \binom{N}{3}, 3! \binom{N}{4}, \dots) = \frac{n!}{k!} \sum_{\ell=k}^{k+n} \mu_\ell(n,k;N) \, N^\ell
\end{equation*}
and the representation 
\begin{equation*}
\mu_\ell(n,k) = \mathop{\sum_{n_1=1}^{n} \cdots \sum_{n_k=1}^{n}}_{n_1+\cdots+n_k=n} \mathop{\sum_{\ell_1=1}^{n_1+1} \cdots \sum_{\ell_k=1}^{n_k+1}}_{\ell_1+\cdots+\ell_k=\ell} \frac{\Stirlings(n_1+1, \ell_1) \cdots \Stirlings(n_k+1, \ell_k)}{(n_1+1)! \cdots (n_k+1)!}.
\end{equation*}

It is well-known that $(-1)^{n-k} \Stirlings(n,k) > 0$ for $n \geq k$. Thus
\begin{equation*}
\prod_{\nu=1}^k \Stirlings(n_\nu+1, \ell_\nu) = \prod_{\nu=1}^k (-1)^{n_\nu + 1 - \ell_\nu} \left| \Stirlings(n_\nu+1, \ell_\nu) \right| = (-1)^{n_1 + \cdots + n_k + k - \ell_1 - \dots - \ell_k} \left| \prod_{\nu=1}^k \Stirlings(n_\nu+1, \ell_\nu) \right| 
\end{equation*}
and from $n_1 + \cdots + n_k = n$ and $\ell_1 + \cdots + \ell_k = \ell$ it follows that
\begin{equation*}
\BellB_{n,k}(1! \binom{N}{2}, 2! \binom{N}{3}, 3! \binom{N}{4}, \dots) = \frac{n!}{k!} \sum_{\ell=k}^{k+n} (-1)^{n-\ell+k} \left| \mu_\ell(n,k) \right| N^\ell.
\end{equation*}

By definition $(-1)^{n-k} \Stirlings(n,k)$ counts the permutations of ${1,2, \dots,n}$ with exactly $k$ cycles. With this combinatorial interpretation in mind we can write
\begin{equation*}
\BellB_{n,k}(1! \binom{N}{2}, 2! \binom{N}{3}, 3! \binom{N}{4}, \dots) = \frac{n!}{k! (n+k)!} \sum_{\ell=k}^{k+n} (-1)^{n-\ell+k} H_\ell(n,k) \, N^\ell,
\end{equation*}
where the coefficients $H_\ell(n,k)$ (which do not depend on $N$) are given in the lemma.
\end{proof}


\begin{lem} \label{lem:aux.7b}
Let $n$ ($\geq 1$) and $k$ ($\geq 1$) be integers. Then
\begin{align*}
H_{n+k}(n,n) &= 2^{-n} (2n)! \binom{n}{k}, \\
H_{n+k}(n,k) &= \mathop{\sum_{n_1=2}^{n+1} \cdots \sum_{n_k=2}^{n+1}}_{n_1+\cdots+n_k=n+k} \binom{n+k}{n_1,n_2,\dots,n_k} = \frac{k! (n+k)!}{n!} \BellB_{n,k}( \frac{1}{2}, \frac{1}{3}, \dots).
\end{align*}
\end{lem}

\begin{proof}
By Lemma~\ref{lem:aux.7},
\begin{equation*}
H_{n+k}(n,n) = \mathop{\sum_{n_1=1}^{n} \cdots \sum_{n_n=1}^{n}}_{n_1+\cdots+n_n=n} \binom{2n}{n_1+1,n_2+1,\dots,n_n+1} \mathop{\sum_{\ell_1=1}^{n_1+1} \cdots \sum_{\ell_n=1}^{n_n+1}}_{\ell_1+\cdots+\ell_n=n+k} \left| \Stirlings(n_1+1, \ell_1) \right| \cdots \left| \Stirlings(n_n+1, \ell_n) \right|.
\end{equation*}
It follows that $n_1 = \cdots = n_n = 1$. Hence, the multinomial coefficient evaluates as $2^{-n}(2n)!$ and 
\begin{equation*}
H_{n+k}(n,n) = 2^{-n}(2n)! \mathop{\sum_{\ell_1=0}^{1} \cdots \sum_{\ell_n=0}^{1}}_{\ell_1+\cdots+\ell_n=k} \left| \Stirlings(2, \ell_1+1) \right| \cdots \left| \Stirlings(2, \ell_n+1) \right|.
\end{equation*}
Since $\Stirlings(2,1) = -1$ and $\Stirlings(2,2) = 1$, the sum counts the number of $n$-tuples $(\ell_1, \dots, \ell_n)$ with precisely $k$ of the entries being $1$ which is equal to $\binom{n}{k}$. The result follows.

Similarly
\begin{equation*}
\begin{split}
H_{n+k}(n,k) 
&= \mathop{\sum_{n_1=1}^{n} \cdots \sum_{n_k=1}^{n}}_{n_1+\cdots+n_k=n} \binom{n+k}{n_1+1,n_2+1,\dots,n_k+1}  \mathop{\sum_{\ell_1=1}^{n_1+1} \cdots \sum_{\ell_k=1}^{n_k+1}}_{\ell_1+\cdots+\ell_k=n+k} \left| \Stirlings(n_1+1, \ell_1) \right| \cdots \left| \Stirlings(n_k+1, \ell_k) \right|.
\end{split}
\end{equation*}
The inner sum reduces to a single term with $\ell_1 = n_1 + 1$, \dots, $\ell_k = n_k + 1$ which equals $1$. Thus
\begin{equation*}
H_{n+k}(n,k) = \mathop{\sum_{n_1=1}^{n} \cdots \sum_{n_k=1}^{n}}_{n_1+\cdots+n_k=n} \binom{n+k}{n_1+1,n_2+1,\dots,n_k+1} = \mathop{\sum_{n_1=2}^{n+1} \cdots \sum_{n_k=2}^{n+1}}_{n_1+\cdots+n_k=n+k} \binom{n+k}{n_1,n_2,\dots,n_k}.
\end{equation*}
On the other hand, by Lemma~\ref{lem:aux.7}, $\frac{n! H_{n+k}(n,k)}{k! (n+k)!} $ 
is the coefficient of the highest power of $N$ in
\begin{equation*}
\BellB_{n,k}(1! \binom{N}{2}, 2! \binom{N}{3}, 3! \binom{N}{4}, \dots).
\end{equation*}
From
\begin{equation*}
\nu ! \binom{N}{\nu+1} = \frac{1}{\nu+1} N^{\nu+1} \left( 1 + \mathcal{O}(N^{-1}) \right) \qquad \text{as $N \to \infty$}
\end{equation*}
and the identities \eqref{eq:BellB.id.1}, we obtain (as $N \to \infty$)
\begin{equation*}
\BellB_{n,k}(1! \binom{N}{2}, 2! \binom{N}{3}, 3! \binom{N}{4}, \dots) = N^{n+k} \, \BellB_{n,k}( \frac{1 + \mathcal{O}(N^{-1})}{2}, \frac{1 + \mathcal{O}(N^{-1})}{3}, \dots). 
\end{equation*}
It follows that
\begin{equation*}
\frac{n!}{k! (n+k)!} H_{n+k}(n,k) = \BellB_{n,k}( \frac{1}{2}, \frac{1}{3}, \dots), \qquad n \geq k \geq 1. 
\end{equation*}
\end{proof}

Lemma~\ref{lem:aux.7} allows us to recast the representation of $f_q^{(n)}(1)$ given in Lemma~\ref{lem:aux.6}.
\begin{cor} \label{cor:aux.8}
Let $q$, $n$ and $N$ be positive integer. Then (with $H_{\ell+k}(n,k)$ of Lemma~\ref{lem:aux.7})
\begin{equation*}
\frac{f_q^{(n)}(1)}{n!} = \sum_{\ell=0}^{n} (-1)^\ell G(n,\ell,q) \, N^{\ell-q}, \qquad G(n,\ell,q) = \sum_{k=1}^n \frac{(-1)^{n-k}\Pochhsymb{q}{k}}{k! (n+k)!} H_{\ell+k}(n,k).
\end{equation*}
\end{cor}

\begin{proof}
From Lemmas~\ref{lem:aux.6} and \ref{lem:aux.7}
\begin{align*}
\frac{f_q^{(n)}(1)}{n!} 
&= \frac{1}{n!} \sum_{k=1}^n (-1)^k \Pochhsymb{q}{k} N^{-q-k} \, \BellB_{n,k}(1! \binom{N}{2}, 2! \binom{N}{3}, 3! \binom{N}{4}, \dots) \\
&= \frac{1}{n!} \sum_{k=1}^n (-1)^k \Pochhsymb{q}{k} N^{-q-k} \, \frac{n!}{k! (n+k)!} \sum_{\ell=k}^{k+n} (-1)^{n-\ell+k} H_\ell(n,k) \, N^\ell. \\
\end{align*}
The result follows by reordering terms.
\end{proof}

\begin{lem} \label{cor:aux.8b}
Let $q$ be a positive integer. Then one has the generating function relation
\begin{equation} \label{eq:generating.fcn.relation}
\sum_{n=1}^\infty (-1)^n G(n,n,q) \, x^n = \left( \frac{e^x-1}{x} \right)^{-q}-1.
\end{equation}
\end{lem}

\begin{proof}
By Corollary~\ref{cor:aux.8} and Lemma~\ref{lem:aux.7b}
\begin{align*}
(-1)^n G(n,n,q) = \sum_{k=1}^n \frac{(-1)^{k}\Pochhsymb{q}{k}}{k! (n+k)!} H_{n+k}(n,k) = \frac{1}{n!} \sum_{k=1}^n (-1)^{k}\Pochhsymb{q}{k} \BellB_{n,k}(1/2, 1/3, \dots).
\end{align*}
Application of the identities (cf. Fa\`{a} di Bruno's differentiation formula)
\begin{equation*}
g(f(x)) = \sum_{n=1}^\infty \frac{\sum_{k=1}^n b_k \, \BellB_{n,k}(a_1, a_2, \dots)}{n!} x^n, \qquad f(x) = \sum_{n=1}^\infty \frac{a_n}{n!} x^n, \quad g(y) = \sum_{n=1}^\infty \frac{b_n}{n!} y^n,
\end{equation*}
to the functions
\begin{equation*}
f(x) = \sum_{n=1}^\infty \frac{x^n}{(n+1)!} = \frac{e^x-1-x}{x}, \qquad g(y) = \sum_{n=1}^\infty \frac{(-1)^n \Pochhsymb{q}{n}}{n!} y^n = \left( 1 + y \right)^{-q} - 1
\end{equation*}
yields $\sum_{n=1}^\infty (-1)^n G(n,n,q) \, x^n = g(f(x)) = \left( \frac{e^x-1}{x} \right)^{-q}-1$.
\end{proof}

With these preparations we are ready to give the following proof.

\begin{proof}[Proof of Theorem~\ref{thm:even.positive.s}]
We rearrange $\mathcal{M}_{2m}(N; r)$ in Theorem~\ref{thm:2} w.r.t. powers of $(1-r)$; i.e., 
\begin{align*}
&\mathcal{M}_{2m}(N; r) + N \left( 1 - r \right)^{-2m}= - \frac{N^2}{\gammafcn(m)} \sum_{k=0}^{m-1} \frac{\gammafcn(2m-k-1)}{\gammafcn(m-k) k!} \Pochhsymb{1-m}{k} \left( 1 + r \right)^{k+1-2m} \left( 1 - r \right)^{k+1-2m} \\
&\phantom{===}+ 2 \frac{N^2}{\gammafcn(m)} \sum_{k=0}^{m-1} \sum_{q=0}^{k} \frac{\gammafcn(2m-k-1)}{\gammafcn(m-k) k!} g(k,q;N,1-m) \left( \frac{1 - r^N}{1 - r} \right)^{-q-1} \left( 1 + r \right)^{k+1-2m} \left( 1 - r \right)^{k-q-2m}.
\end{align*}
Note that $k+1-2m<0$. Using Taylor expansion (cf. Lemma~\ref{lem:aux.6}, since $f_q(r) = ( 1 + r )^{-q}$ if $N = 2$)
\begin{equation*}
\left( 1 + r \right)^{-q} = \sum_{n=0}^m \frac{\Pochhsymb{q}{n}}{n!} 2^{-q-n} \left( 1 - r \right)^n + \tilde{\mathcal{R}}_m(r), \qquad \left| \tilde{\mathcal{R}}_m(r) \right| \leq \frac{\Pochhsymb{q}{m+1}}{(m+1)!} \left( 1 - r \right)^{m+1},
\end{equation*}
one gets ($N$ is assumed to be fixed)
\begin{align*}
&\mathcal{M}_{2m}(N; r) + N \left( 1 - r \right)^{-2m} = - \frac{N^2}{\gammafcn(m)} \sum_{k=0}^{m-1} \frac{\gammafcn(2m-k-1)}{\gammafcn(m-k) k!} \Pochhsymb{1-m}{k} \\
&\phantom{===\pm}\times \sum_{n=0}^{2m-1-k} \frac{\Pochhsymb{2m-1-k}{n}}{n!} 2^{k-n+1-2m} \left( 1 - r \right)^{k+n+1-2m} + \mathcal{O}((1-r)) \\
&\phantom{===}+ 2 \frac{N^2}{\gammafcn(m)} \sum_{k=0}^{m-1} \sum_{q=0}^{k} \frac{\gammafcn(2m-k-1)}{\gammafcn(m-k) k!} g(k,q;N,1-m) \left( \frac{1 - r^N}{1 - r} \right)^{-q-1} \\
&\phantom{===\pm}\times \sum_{n=0}^{2m+q-k} \frac{\Pochhsymb{2m-1-k}{n}}{n!} 2^{k-n+1-2m} \left( 1 - r \right)^{n+k-q-2m} + \mathcal{O}((1-r)) \\
&\phantom{==}= - \frac{N^2}{\gammafcn(m)} \sum_{k=0}^{m-1} \sum_{n=0}^{2m-1-k} \frac{\gammafcn(n+2m-k-1)}{\gammafcn(m-k) k! n!} \Pochhsymb{1-m}{k} 2^{k+1-2m-n} \left( 1 - r \right)^{n+k+1-2m} \\
&\phantom{===}+ 2 \frac{N^2}{\gammafcn(m)} \sum_{k=0}^{m-1} \sum_{q=0}^{k} \sum_{n=0}^{2m+q-k} \frac{\gammafcn(n+2m-k-1)}{\gammafcn(m-k) k! n!} g(k,q;N,1-m) \left( \frac{1 - r^N}{1 - r} \right)^{-q-1} \\
&\phantom{===\pm}\times 2^{k+1-2m-n} \left( 1 - r \right)^{n+k-q-2m} 
+ \mathcal{O}((1-r)) \qquad \text{as $r \to 1^-$.}
\end{align*}
By Lemma~\ref{lem:aux.6}
\begin{align*}
&\mathcal{M}_{2m}(N; r) + N \left( 1 - r \right)^{-2m} \\
&\phantom{=}=- \frac{N^2}{\gammafcn(m)} \sum_{k=0}^{m-1} \sum_{n=0}^{2m-1-k} \frac{\gammafcn(n+2m-k-1)}{\gammafcn(m-k) k! n!} \Pochhsymb{1-m}{k} 2^{k+1-2m-n} \left( 1 - r \right)^{n+k+1-2m} \\
&\phantom{===}+ 2 \frac{N^2}{\gammafcn(m)} \sum_{k=0}^{m-1} \sum_{q=0}^{k} \sum_{n=0}^{2m+q-k} \frac{\gammafcn(n+2m-k-1)}{\gammafcn(m-k) k! n!} g(k,q;N,1-m) N^{-1-q} 2^{k+1-2m-n} \left( 1 - r \right)^{n+k-q-2m} \\
&\phantom{===}+ 2 \frac{N^2}{\gammafcn(m)} \sum_{k=0}^{m-1} \sum_{q=0}^{k} \sum_{n=0}^{2m+q-k} \sum_{\nu=1}^{2m+q-k-n} \frac{\gammafcn(n+2m-k-1)}{\gammafcn(m-k) k! n!} g(k,q;N,1-m) \frac{(-1)^\nu f_{q+1}^{(\nu)}(1)}{\nu!} \\
&\phantom{===\pm}\times 2^{k+1-2m-n} \left( 1 - r \right)^{\nu+n+k-q-2m} + \mathcal{O}((1-r)) \qquad \text{as $r \to 1^-$.}
\end{align*}
The power $(1-r)^{-2m}$ only appears in the triple sum when $q = k$ and $n = 0$. Its coefficient is 
\begin{equation*}
c_{-2m} \DEF 2 \frac{N^2}{\gammafcn(m)} \sum_{k=0}^{m-1} \frac{\gammafcn(2m-k-1)}{\gammafcn(m-k) k!} g(k,k;N,1-m) N^{-1-k} 2^{k+1-2m}.
\end{equation*}
By the formulas in Theorem~\ref{thm:2} and the properties of Stirling and Eulerian numbers (cf. \eqref{eq:b.n.0})
\begin{equation*}
g(k,k; N, 1 - m) = \Stirlings(k,k;1-m) b(k,0) \, N^k = s(k,k) \,  N^k \sum_{j=0}^k \Eulerian{k}{j} = s(k,k) k! N^k = k! N^k.
\end{equation*}
Thus, by \cite[Eq.~4.2.3.10]{PrBrMa1986I},
\begin{equation*}
c_{-2m} = N 2^{2-2m} \sum_{k=0}^{m-1} \frac{\gammafcn(2m-k-1)}{\Gamma(m)\gammafcn(m-k)} 2^k = N 2^{2-2m} \sum_{k=0}^{m-1} \binom{2m-2-k}{m-1} 2^k = N.
\end{equation*}
Now the term $N (1-r)^{-2m}$ can be cancelled in above relation. 
First, we reorder the terms 
\begin{align*}
&\mathcal{M}_{2m}(N; r) + N \left( 1 - r \right)^{-2m} \\
%
&\phantom{==}= - \frac{N^2}{\gammafcn(m)} \sum_{k=0}^{m-1} \sum_{n=0}^{m+k} \frac{\gammafcn(2(m+k)-n)}{\gammafcn(m-k) k! (m+k-n)!} \Pochhsymb{1-m}{m-1-k} 2^{n-2(m+k)} \left( 1 - r \right)^{-n} \\
&\phantom{===}+ 2 \frac{N^2}{\gammafcn(m)} \sum_{k=0}^{m-1} \sum_{q=0}^{k} \sum_{n=0}^{2m+q-k} \frac{\gammafcn(2(2m-1-k)+q-n+1)}{\gammafcn(m-k) k! (2m+q-k-n)!} g(k,q;N,1-m) N^{-1-q} \\
&\phantom{===\pm}\times 2^{n-q-1-2(2m-1-k)} \left( 1 - r \right)^{-n} \\
&\phantom{===}+ 2 \frac{N^2}{\gammafcn(m)} \sum_{k=0}^{m-1} \sum_{q=0}^{k} \sum_{n=0}^{2m+q-k} \sum_{\nu=0}^{n-1} \frac{\gammafcn(2(2m-1-k)+q-n+1)}{\gammafcn(m-k) k! (2m+q-k-n)!} g(k,q;N,1-m) \\
&\phantom{===\pm}\times \frac{(-1)^{n-\nu} f_{q+1}^{(n-\nu)}(1)}{(n-\nu)!} 2^{n-q-1-2(2m-1-k)} \left( 1 - r \right)^{-\nu} + \mathcal{O}((1-r)) 
, \qquad \text{as $r \to 1^-$.}
\end{align*}
Since $\lim_{r\to1^-} \mathcal{M}_{2m}(N; r) = \mathcal{M}_{2m}(N)$ and $\mathcal{M}_{2m}(N)$ exists and is finite for all positive integers $m$ and $N$ with $N\geq2$ (and $N (1-r)^{-2m}$ can be cancelled), it follows that the cumulative coefficient of a negative power of $(1-r)$ is zero and we are left with
\begin{align*}
&\mathcal{M}_{2m}(N) 
= - \frac{N^2}{\gammafcn(m)} \sum_{k=0}^{m-1} \frac{\gammafcn(2(m+k))}{\gammafcn(m-k) k! (m+k)!} \Pochhsymb{1-m}{m-1-k} 2^{-2(m+k)} \\
&\phantom{=}+ 2 \frac{N^2}{\gammafcn(m)} \sum_{k=0}^{m-1} \sum_{q=0}^{k} \frac{\gammafcn(2(2m-1-k)+q+1)}{\gammafcn(m-k) k! (2m+q-k)!} g(k,q;N,1-m) N^{-1-q} 2^{-q-1-2(2m-1-k)} \\
&\phantom{=}+ 2 \frac{N^2}{\gammafcn(m)} \sum_{k=0}^{m-1} \sum_{q=0}^{k} \sum_{n=0}^{2m+q-k} \frac{\gammafcn(2(2m-1-k)+q-n+1)}{\gammafcn(m-k) k! (2m+q-k-n)!} g(k,q;N,1-m) \\
&\phantom{===\pm}\times  \frac{(-1)^{n} f_{q+1}^{(n)}(1)}{n!} 2^{n-q-1-2(2m-1-k)}.
\end{align*}
Substituting the expressions for $g(k,q;N,1-m)$ (Theorem~\ref{thm:2}) and $f_{q+1}^{(n)}(1)$ (Cor.~\ref{cor:aux.8}), we get
\begin{equation}
\begin{split} \label{eq:formula.01}
&\mathcal{M}_{2m}(N) 
= - \frac{N^2}{\gammafcn(m)} \sum_{k=0}^{m-1} \frac{\gammafcn(2(m+k))}{\gammafcn(m-k) k! (m+k)!} \Pochhsymb{1-m}{m-1-k} 2^{-2(m+k)} \\
&\phantom{=}+ N \sum_{k=0}^{m-1} \sum_{q=0}^{k} \sum_{p=0}^{k-q} X(k,q,p;m) N^{p} + N \sum_{k=0}^{m-1} \sum_{q=0}^{k} \sum_{n=0}^{2m+q-k} \sum_{p=0}^{k-q} \sum_{\ell=0}^{n} X(k,q,p,n,\ell;m) \, N^{\ell+p},
\end{split}
\end{equation}
where the primary coefficients are given by
\begin{align*}
X(k,q,p;m) &\DEF \frac{2}{\gammafcn(m)} \frac{\gammafcn(2(2m-1-k)+q+1)}{\gammafcn(m-k) k! (2m+q-k)!} \frac{(-1)^{p} s(k,p+q;1-m) b(p+q,p)}{2^{q+1+2(2m-1-k)}} \\
X(k,q,p,n,\ell;m) &\DEF \frac{2}{\gammafcn(m)} \frac{\gammafcn(2(2m-1-k)+q-n+1)}{\gammafcn(m-k) k! (2m+q-k-n)!} \frac{(-1)^{p+n+\ell} s(k,p+q;1-m) b(p+q,p)}{2^{q+1+2(2m-1-k)-n}} \\
&\phantom{=\pm}\times G(n,\ell,q+1)
\end{align*}
which depend on $G(n,\ell,q)$ given in Corollary~\ref{cor:aux.8} (see also Lemma~\ref{lem:aux.7}) and the secondary coefficients
\begin{align*}
\Stirlings(n, \ell; y) &= \sum_{k=\ell}^n \binom{k}{\ell} \Stirlings(n,k) \left( y+n-1 \right)^{k-\ell}, \qquad b(p+q,p) = \sum_{j=0}^q \Eulerian{p+q}{j} \binom{p+q-j}{p}.
\end{align*}

We reorder the terms in \eqref{eq:formula.01} w.r.t. powers of $N$. The first  part is
\begin{equation*}
\sum_{k=0}^{m-1} \sum_{q=0}^{k} \sum_{p=0}^{k-q} X(k,q,p;m) N^{p} = \sum_{k=0}^{m-1} \sum_{p=0}^{k} \sum_{q=0}^{k-p} X(k,q,p;m) N^{p} = \sum_{p=0}^{m-1} \sum_{k=p}^{m-1} \sum_{q=0}^{k-p} X(k,q,p;m) N^{p}.
\end{equation*}
Letting $[ \cdots ] = X(k,k-q,p,n,\ell;m) \, N^{\ell+p}$, the second part is
\begin{align*}
&\sum_{k=0}^{m-1} \sum_{q=0}^{k} \sum_{n=0}^{2m+q-k} \sum_{p=0}^{k-q} \sum_{\ell=0}^{n} X(k,q,p,n,\ell;m) \, N^{\ell+p} = \sum_{k=0}^{m-1} \sum_{q=0}^{k} \sum_{p=0}^{q} \sum_{n=0}^{2m-q} \sum_{\ell=0}^{n} \Big[ \cdots \Big] \\
&\phantom{=}= \sum_{k=0}^{m-1} \sum_{q=0}^{k} \sum_{p=0}^{q} \sum_{\ell=0}^{2m-q} \sum_{n=\ell}^{2m-q}  \Big[ \cdots \Big] = \sum_{k=0}^{m-1} \sum_{p=0}^{k} \sum_{q=p}^{k} \sum_{\ell=0}^{2m-q} \sum_{n=\ell}^{2m-q}  \Big[ \cdots \Big] = \sum_{p=0}^{m-1} \sum_{k=p}^{m-1} \sum_{q=p}^{k} \sum_{\ell=0}^{2m-q} \sum_{n=\ell}^{2m-q}  \Big[ \cdots \Big] \\
&\phantom{=}= \sum_{p=0}^{m-1} \sum_{q=p}^{m-1} \sum_{\ell=0}^{2m-q} \sum_{k=q}^{m-1} \sum_{n=\ell}^{2m-q}  \Big[ \cdots \Big] = \left[ \sum_{p=0}^{m-1} \sum_{\ell=0}^{m} \sum_{q=p}^{m-1} \sum_{k=q}^{m-1} \sum_{n=\ell}^{2m-q} + \sum_{p=0}^{m-1} \sum_{\ell=m+1}^{2m-p} \sum_{q=p}^{2m-\ell} \sum_{k=q}^{m-1} \sum_{n=\ell}^{2m-q} \right]  \Big[ \cdots \Big] \\
&\phantom{=}= \left[ \sum_{\nu=0}^{m-1} \sum_{p=0}^{\nu} \sum_{q=p}^{m-1} \sum_{k=q}^{m-1} \sum_{n=\nu-p}^{2m-q} + \sum_{\nu=m}^{2m-1} \sum_{p=\nu-m}^{m-1} \sum_{q=p}^{m-1} \sum_{k=q}^{m-1} \sum_{n=\nu-p}^{2m-q} \right] X(k,k-q,p,n,\nu-p;m) \, N^{\nu} \\
&\phantom{==}+ \sum_{\nu=m+1}^{2m} \sum_{p=0}^{\nu-m-1} \sum_{q=p}^{2m-\nu+p} \sum_{k=q}^{m-1} \sum_{n=\nu-p}^{2m-q} X(k,k-q,p,n,\nu-p;m) \, N^{\nu}.
\end{align*}

Putting everything together, we arrive at 
\begin{align*}
\mathcal{M}_{2m}(N) 
&= - \frac{N^2}{\gammafcn(m)} \sum_{k=0}^{m-1} \frac{\gammafcn(2(m+k))  \Pochhsymb{1-m}{m-1-k}}{\gammafcn(m-k) k! (m+k)!} 2^{-2(m+k)} + N \sum_{p=0}^{m-1} \sum_{k=p}^{m-1} \sum_{q=0}^{k-p} X(k,q,p;m) N^{p} \\
&\phantom{=}+ N \left[ \sum_{\nu=0}^{m-1} \sum_{p=0}^{\nu} \sum_{q=p}^{m-1} \sum_{k=q}^{m-1} \sum_{n=\nu-p}^{2m-q} + \sum_{\nu=m}^{2m-1} \sum_{p=\nu-m}^{m-1} \sum_{q=p}^{m-1} \sum_{k=q}^{m-1} \sum_{n=\nu-p}^{2m-q} \right] X(k,k-q,p,n,\nu-p;m) \, N^{\nu} \\
&\phantom{=}+ N \sum_{\nu=m+1}^{2m} \sum_{p=0}^{\nu-m-1} \sum_{q=p}^{2m-\nu+p} \sum_{k=q}^{m-1} \sum_{n=\nu-p}^{2m-q} X(k,k-q,p,n,\nu-p;m) \, N^{\nu}.
\end{align*}
Direct computations (with the help of Mathematica) give the expected results (cf. \eqref{exp2xct})
\begin{equation*}
\mathcal{M}_{2}(N) = \frac{N^3}{12} - \frac{N}{12}, \qquad \mathcal{M}_{4}(N) = \frac{N^5}{720} + \frac{10}{720} N^3 - \frac{11}{720} N.
\end{equation*}
In general, $\mathcal{M}_{2m}(N) = \sum_{\nu = 0}^{2m} \beta_\nu(m) \, N^{1+\nu}$, $m \geq 2$,
where for $\nu = 0, 1, \dots, m - 1$ and $\nu \neq 1$,
\begin{subequations} \label{eq:beta.coeff.s}
\begin{align}
\beta_\nu(m) &\DEF \sum_{k=\nu}^{m-1} \sum_{q=0}^{k-\nu} X(k,q,\nu;m) + \sum_{p=0}^{\nu} \sum_{q=p}^{m-1} \sum_{k=q}^{m-1} \sum_{n=\nu-p}^{2m-q} X(k,k-q,p,n,\nu-p;m); 
\intertext{for $\nu = m$ and $m \neq 1$,}
\beta_m(m) &\DEF \sum_{p=0}^{m-1} \sum_{q=p}^{m-1} \sum_{k=q}^{m-1} \sum_{n=m-p}^{2m-q} X(k,k-q,p,n,m-p;m); 
\intertext{for $\nu = m + 1, \dots, 2m - 1$ (and $m \geq 2$),}
\begin{split}
\beta_\nu(m) &\DEF \left[ \sum_{p=\nu-m}^{m-1} \sum_{q=p}^{m-1} \sum_{k=q}^{m-1} \sum_{n=\nu-p}^{2m-q} + \sum_{p=0}^{\nu-m-1} \sum_{q=p}^{2m-\nu+p} \sum_{k=q}^{m-1} \sum_{n=\nu-p}^{2m-q} \right] X(k,k-q,p,n,\nu-p;m);
\end{split}
\intertext{for $\nu = 2m$ (and $m \geq 2$),}
\beta_{2m}(m) &\DEF \sum_{p=0}^{m-1} \sum_{k=p}^{m-1} X(k,k-p,p,2m-p,2m-p;m); \label{eq:beta.coeff.s-A}
\intertext{and for $\nu = 1$ (and $m \geq 2$),}
\begin{split}
\beta_1(m) &= - \frac{1}{\gammafcn(m)} \sum_{k=0}^{m-1} \frac{\gammafcn(2(m+k)) \Pochhsymb{1-m}{m-1-k}}{\gammafcn(m-k) k! (m+k)!} 2^{-2(m+k)} \\
&\phantom{=}+ \sum_{k=1}^{m-1} \sum_{q=0}^{k-1} X(k,q,1;m) + \sum_{p=0}^{1} \sum_{q=p}^{m-1} \sum_{k=q}^{m-1} \sum_{n=1-p}^{2m-q} X(k,k-q,p,n,1-p;m).
\end{split}
\end{align}
\end{subequations}
The result follows.
\end{proof}


\begin{proof}[Proof of Corollary~\ref{cor:4}]
By definition of $X(k,q,p,n,\ell;m)$ and \eqref{eq:beta.coeff.s-A}, we get 
\begin{align*}
\beta_{2m}(m) 
&= \sum_{p=0}^{m-1} \sum_{k=p}^{m-1} (-1)^p \frac{2^{2+k-2m} \gammafcn(2m-1-k)}{\gammafcn(m) k! \gammafcn(m-k)} \Stirlings(k,k; 1 - m) b(k,p) G(2m-p,2m-p,k-p+1).
\end{align*}
By the definitions of $\Stirlings(n,k;x)$ (Lemma~\ref{lem:aux.2}) and $b(n,\ell)$ (Lemma~\ref{lem:aux.3}), we have
\begin{equation*}
\Stirlings(k,k; 1 - m) = \Stirlings(k,k) = 1, \qquad b(k,p) = \sum_{j=0}^{k-p} \Eulerian{k}{j} \binom{k-j}{p}
\end{equation*}
and the $G(n,n,q)$ are defined by means of~\eqref{eq:generating.fcn.relation}. 
The result follows after rearranging terms. 
\end{proof}


{\bf Acknowledgement:} The research was supported, in part, by an APART-Fellowship of the Austrian Academy of Sciences. The author is grateful for the hospitality of School of Mathematics and Statistics at University of New South Wales where part of this research was conducted.

\appendix

\section{Auxiliary results}
\label{sec:appendix}

\begin{lem} \label{lem:aux.2}
Let $n$ be a nonnegative integer. Then
\begin{equation*}
\Pochhsymb{x+y}{n} = \sum_{\ell=0}^n \Stirlings(n, \ell; y) \, x^\ell, \qquad \Stirlings(n, \ell; y) \DEF \sum_{k=\ell}^n \binom{k}{\ell} \Stirlings(n,k) \left( y+n-1 \right)^{k-\ell} = \sum_{k=\ell}^n \binom{n}{k} \usgnStirlingS{k}{\ell} \Pochhsymb{y}{n-k},
\end{equation*}
where $\usgnStirlingS{n}{k}$ is the unsigned Stirling number of the first kind.
\end{lem}

\begin{proof}
Using \eqref{eq:Stirling.s} and the Binomial theorem, one gets
\begin{align*}
\Pochhsymb{x+y}{n} 
&= \sum_{k=0}^n \Stirlings(n,k) \left( x+y+n-1 \right)^k = \sum_{k=0}^n \sum_{\ell=0}^k \Stirlings(n,k) \binom{k}{\ell} x^\ell \left( y+n-1 \right)^{k-\ell}.
\end{align*}
The first formula for $\Stirlings(n, \ell; y)$ follows after reordering the sum. The second one follows from 
\begin{equation*}
\Pochhsymb{a+b}{n} = \sum_{j=0}^\infty \binom{n}{j} \Pochhsymb{a}{n-j} \Pochhsymb{b}{j}, \qquad \Pochhsymb{x}{n} = \sum_{k=0}^n \usgnStirlingS{n}{k} x^k.
\end{equation*}
\end{proof}

Note that 
\begin{equation} \label{eq:s.n.0.y}
\Stirlings(n, 0; y) = \sum_{k=0}^n \Stirlings(n,k) \left( y+n-1 \right)^{k} = \Pochhsymb{y+n-1+1-n}{n} = \Pochhsymb{y}{n}.
\end{equation}
%

We need the {\em polylogarithm function} defined by
\begin{equation} \label{eq:polylogarithm}
\PolyLog_\nu(z) = \sum_{k=1}^\infty \frac{z^k}{k^\nu}, \qquad | z | < 1.
\end{equation}
In particular, one has for $\nu=-n$ and $n$ a positive integer the relations (cf. \cite{Mi1983}, \cite{MathWorldPolylogarithm2009})
\begin{equation} \label{eq:polylogarithm.identities}
\PolyLog_{-n}(z) = \sum_{k=1}^\infty k^n z^k = \frac{1}{\left( 1 - z \right)^{n+1}} \sum_{j=0}^{n-1} \Eulerian{n}{j} z^{n-j}, \qquad |z|<1.
\end{equation}

\begin{lem} \label{lem:aux.3}
Let $n$ be a positive integer. Then
\begin{equation*}
\PolyLog_{-n}(z) = \frac{n! + \sum_{\ell=1}^n (-1)^\ell b(n,\ell) \left( 1 - z \right)^\ell}{\left( 1 - z \right)^{n+1}}, \qquad b(n,\ell) \DEF \sum_{j=0}^{n-\ell} \Eulerian{n}{j} \binom{n-j}{\ell}.
\end{equation*}
\end{lem}

\begin{proof}
By \eqref{eq:polylogarithm.identities} and the Binomial theorem
\begin{equation*}
\PolyLog_{-n}(z) = \frac{1}{\left( 1 - z \right)^{n+1}} \sum_{j=0}^{n-1} \Eulerian{n}{j} \left[1 - \left( 1 - z \right) \right]^{n-j} = \frac{1}{\left( 1 - z \right)^{n+1}} \sum_{j=0}^{n-1} \sum_{\ell=0}^{n-j} \Eulerian{n}{j} \binom{n-j}{\ell} (-1)^\ell \left( 1 - z \right)^\ell.
\end{equation*}
Reordering of the sum yields
\begin{equation*}
\PolyLog_{-n}(z) = \frac{1}{\left( 1 - z \right)^{n+1}} \left\{ \sum_{j=0}^{n-1} \Eulerian{n}{j} + \sum_{\ell=1}^{n} \left[ \sum_{j=0}^{n-\ell} \Eulerian{n}{j} \binom{n-j}{\ell} \right] (-1)^\ell \left( 1 - z \right)^\ell \right\}.
\end{equation*}
The sum over the Eulerian numbers gives $n!$ (cf. \cite[Eq.~26.14.10]{Olver:2010:NHMF}). The result follows.
\end{proof}

Note that (using $\Eulerian{q}{q} = 0$ for $q \geq 1$)
\begin{equation} \label{eq:b.n.0}
b(n,0) = \sum_{j=0}^n \Eulerian{n}{j} = \sum_{j=0}^{n-1} \Eulerian{n}{j} = n!.
\end{equation}

\begin{lem} \label{lem:aux.4}
Let $k$ be a non-negative integer and $a,b$ complex numbers. Then 
\begin{equation*}
\sum_{\nu=-\infty}^\infty \Pochhsymb{| \nu | a + b}{k} z^{| \nu |} = - \Pochhsymb{b}{k} + 2 \sum_{q=0}^{k} g(k,q;a,b) \left( 1 - z \right)^{-q-1}, \qquad | z | < 1,
\end{equation*}
where ($s(n,\ell;y)$ is defined in Lemma~\ref{lem:aux.2} and $b(n,\ell)$ is defined in Lemma~\ref{lem:aux.3})
\begin{equation*}
g(k,q;a,b) \DEF \sum_{p=0}^{k-q} (-1)^{p} s(k,p+q;b) b(p+q,p) a^{p+q}. 
\end{equation*}
\end{lem}

\begin{proof}
Let $f(z)$ denote the inifinte series above. For $k = 0$, one gets
\begin{equation*}
f(z) = \sum_{\nu=-\infty}^\infty z^{| \nu |} = 1 + 2 \sum_{\nu=1}^\infty z^{\nu} = 1 + \frac{2z}{1-z} = \frac{1+z}{1-z}.
\end{equation*}
Let $k \geq 1$. By Lemma~\ref{lem:aux.2}, one has
\begin{equation*}
f(z) = \Pochhsymb{b}{k} + 2 \sum_{\nu=1}^\infty \Pochhsymb{\nu a + b}{k} z^\nu = \Pochhsymb{b}{k} + 2 \sum_{p=0}^k \Stirlings(k, p; b) a^p \sum_{\nu=1}^\infty \nu^p z^\nu =\Pochhsymb{b}{k} + 2 \sum_{p=0}^k \Stirlings(k, p; b) a^p \PolyLog_{-p}(z).
\end{equation*}
By Lemma~\ref{lem:aux.3} and \eqref{eq:Stirling.s} (and using Equations~\eqref{eq:s.n.0.y} and \eqref{eq:b.n.0}), one gets 
\begin{align*}
f(z) 
&= \Pochhsymb{b}{k} + 2 \Stirlings(k, 0; b) \PolyLog_{0}(z) + 2 \sum_{p=1}^k \frac{\Stirlings(k, p; b) a^p p! + \sum_{\ell=1}^p (-1)^\ell b(p,\ell) \left( 1 - z \right)^\ell}{\left( 1 - z \right)^{p+1}} \\
&= \Pochhsymb{b}{k} + 2 \Pochhsymb{b}{k} \frac{z}{1-z} + 2 \sum_{p=1}^k \frac{\Stirlings(k, p; b) p!  a^p}{\left( 1 - z \right)^{p+1}} + 2 \sum_{p=1}^k \sum_{\ell=1}^p \frac{\Stirlings(k, p; b) a^p (-1)^\ell b(p,\ell)}{\left( 1 - z \right)^{p + 1 - \ell}}.
\end{align*}
Simplification of the first three terms and summing up over equal powers of $(1-z)$ while using $\Eulerian{n}{n} = 0$ for $n\geq1$ and \eqref{eq:b.n.0} yields the desired result. 
\end{proof}

\bibliographystyle{abbrv}
\bibliography{bibliography}

\end{document}